\newcommand {\indepc}[2] {#1 ~\bot~ #2}

\newcommand{\rows}{\text{\#rows}}

\newcommand\rank[1]{\|#1\|}

\newcommand{\open}{\mathbb}

\newcommand{\oR}{{\open R}}

\def\P{\mathcal{P}}

\def \fdep{{=\mkern-1.2mu}}

\documentclass{article}
\usepackage{graphicx,breakurl}
\usepackage{amsfonts}
\usepackage{amsmath,hyperref}
\usepackage{amssymb,latexsym,times}
\usepackage{mathtools,amsthm}

\usepackage[shortlabels]{enumitem}

\newtheorem{theorem}{Theorem}
\newtheorem{proposition}[theorem]{Proposition}%
\newtheorem{corollary}[theorem]{Corollary}%
\newtheorem{lemma}[theorem]{Lemma}%
\newtheorem{definition}{Definition}%

\title{Diversity, Dependence and Independence}

\author{{Pietro} {Galliani\footnote{Pietro.Galliani@unibz.it}}\\
{Faculty of Computer Science},\\ {Free University of Bozen-Bolzano}, \\
{Bozen-Bolzano}, 
 {Italy}
\and {Jouko} {V\"a\"an\"anen}\footnote{jouko.vaananen@helsinki.fi}\\
{Department of Mathematics and Statistics},\\ {University of Helsinki}, \\
 {Helsinki}, {Finland}\\
 and\\
{ILLC, Universiteit van Amsterdam}, \\
{ {Amsterdam},  {The Netherlands}}
}

\def\boto{\ \bot\ }

\begin{document}

\maketitle

\begin{abstract}
We propose a very general, unifying framework for the concepts of
dependence and independence. For this purpose, we introduce the
notion of {\em diversity rank}. By means of this diversity rank we identify  {\em total determination} with the inability to create more diversity, and  {\em independence} with the presence of maximum diversity. We show that our theory of dependence and independence   covers a variety of dependence concepts, for example the seemingly unrelated  concepts of linear dependence in algebra and dependence of variables in logic.   
\end{abstract}

\section{Introduction}

The concepts of dependence and independence occur widely in science. The exact study of these concepts has taken place at least in four different contexts:

\begin{itemize}

\item {\bf Mathematics:} Dependence and independence are fundamental concepts in algebra: linear dependence in linear algebra and algebraic dependence in field theory. In both cases independence is defined as the lack of dependence: elements $\{x_1,\ldots,x_n\}$ are {\em independent} if no $x_i$ is dependent on the rest. Whitney \cite{MR1507091} and van der Waerden \cite{MR0002841} pointed out the similarity between these two notions of dependence and proposed axioms that cover both cases. Whitney suggested the name {\em matroid} for the general dependence structure inherent in algebra, giving rise to {\em matroid theory}, nowadays a branch of discrete mathematics.  

\item {\bf Computer science:} {\em Functional} dependence \cite{armstrong} is a fundamental concept of data\-base 
theory. The design and analysis of so called relational databases is often based on a careful study of the functional dependencies between attributes of various parts of the database. The more general  {\em multivalued} dependencies are analogous to what we call independence relations between attributes. 

\item {\bf Statistics and probability theory:} Dependence and independence of events (or random variables) is the basis of probability theory and statistical analysis of data.

\item {\bf Logic:} Dependence of a variable on another is the basic concept in quantification theory. In {\em Dependence Logic} \cite{MR2351449} this concept is separated from quantification, making it possible, as in {\em Independence-Friendly Logic} \cite{MR2807973}, to write formulas with more complicated dependence relations between variables than what first order logic allows. Likewise \emph{Independence Logic} \cite{GV}  extends First Order Logic by an atom $\vec x \boto \vec y$ that states that the tuples of quantified variables $\vec x$ and $\vec y$ are chosen independently, in the sense that every possible choice of $\vec x$ and of $\vec y$ may occur together. These logics -- and the generalization of Tarski's Semantics used for their analysis, commonly called \emph{Team Semantics} -- have lead in the last decade to a considerable amount of research regarding logics augmented by various \emph{notions of dependence and independence}, in the first order case but also in the propositional case \cite{yang2016propositional}, in the modal case \cite{vaananen2008modal,hella2014expressive}, in the temporal case \cite{krebs2017team}, and recently even in probabilistic cases \cite{hyttinen2017logic,durand2018probabilistic,hannula2019facets}.
\end{itemize}

In this paper, we show that these seemingly very diverse notions
of dependence and independence can be captured by a simple, very
general, unifying framework.

Our starting point is very general. Suppose we have a set $M$ of objects. We want to make sense of the concept that a finite subset\footnote{In this work, we will write $x \subseteq_f M$ for ``$x$ is a finite subset of $M$''.} $x\subseteq_f M$ {\bf depends} on another subset $y\subseteq_f M$, or that a subset $x\subseteq_f M$ is {\bf independent} of another subset $y\subseteq_f M$. 
To accomplish this in the most general sense, we  define the concept of the {\bf diversity} of a set $x\subseteq_f M$. A small set has less diversity than a bigger set, hence our diversity function is monotone. Also, the diversity of $x$ arises from properties of the individual elements, hence our diversity function satisfies certain further conditions. The connection between diversity and dependence arises from the idea that dependence reduces diversity, and respectively independence preserves diversity. If $y$ is totally determined by $x$, then adding $y$ to $x$ does not increase the diversity of $x$ at all. On the other hand, if $x$ and $y$ are independent, then putting them together means simply adding the diversities together: nothing is lost, because there is no interaction between $x$ and $y$.

Because of the generality of our approach, according to which $M$ is just a set of objects about which we a priori know nothing, we do not {\em define} the diversity function explicitly, but rather give a few conditions it ought to satisfy. The point is that on the basis of these conditions we can introduce natural notions of {\bf dependence} and {\bf independence} with a variety of applications, in particular the ones mentioned above. 


We will now give an overview of our results. In Section 2 we introduce  the concept of diversity rank $\rank{x}$, a non-negative real number, that the whole paper is about. It is defined  in a completely general setting by means of four axioms for $\rank{x}$. We show that another general and closely related  concept, that of a matroid, is a  special case. Thus our approach generalizes the matroid approach. We use our diversity rank to define two binary relations, namely $=\!\!(x,y)$ (the dependence relation  ``$y$ is totally determined on $x$") and  $x\ \bot\ y$ (the independence relation ``$x$ and $y$ are independent"). Previously dependence and independence relations were studied under various guises in particular contexts such as team semantics, databases, algebraic structures, and statistics. In our approach these concepts are defined in a general setting which covers the special cases mentioned.

 In Section 3 we give examples of diversity rank functions. The most interesting examples are relational diversity, algebraic diversity and entropy. We show that relational diversity does not satisfy submodularity, a property that instead matroids (such as the ones corresponding to algebraic diversity) necessarily satisfy. The entropy of a set of random variables, as we will see, also satisfies the submodularity property; however, it is not necessarily integer, and thus it does not correspond to a matroid either. 
 
 The main benefit of our approach as compared to the matroid approach is that we can develop a theory of diversity which covers matroids, the relational case, and the probabilistic (entropy) case. The relational case is important because it arises naturally in database theory, where our dependence is called functional dependence and our independence is called embedded multivalued dependence. Likewise, the entropy diversity rank function gives rise to the known probability-theoretic notions of independence and functional dependence between (sets of) random variables.
 
 As was observed above, our diversity rank $\rank{x}$ makes it possible to define a  dependence 
relation $=\!\!(x,y)$. In section 4 we use the diversity rank axioms to prove the so-called Armstrong Axioms for this relation.
In Theorem 1 we prove the completeness of Armstrong Axioms in our more general context. Thus the well-known completeness of Armstrong Axioms in team semantics or in the theory of database dependencies is a more general phenomenon. 

In Section 5 we use diversity rank axioms to prove (a simple extension of) the Geiger-Paz-Pearl axioms for the  independence relation $x\ \bot\ y$ derived from a diversity rank $\rank{x}$. In Theorem 2 we prove the completeness of these axioms in our general context.

In Section 6 we address the ``reverse" question, whether every binary relation $=\!\!(x,y)$ satisfying the Armstrong Axioms of dependence arises from a diversity rank function $\rank{x}$. In Theorem 3 we give a positive answer in the case of countable domains. The corresponding question for the Geiger-Paz-Pearl axioms for the independence relation remains open. 

\section{Diversity rank in a general setting}

We now define the concept of {\em diversity rank} in an entirely general setting. We use the notation $xy$ to denote the union $x\cup y$ of  subsets $x$ and $y$ of a fixed set $M$. The following is the key definition of this paper:

\begin{definition}
Suppose $M$ is an arbitrary set. A function $x\mapsto \rank{x}$ from the finite subsets of $M$ to $\oR^+\cup\{0\}$ is called a {\em diversity rank function on $M$} if it satisfies the following conditions for all $x, y, z \subseteq_f M$:

\begin{description}
	\item [\textbf{R1:}] $\rank{\emptyset}=0$;

	\item [\textbf{R2:}] $\rank{x}\le\rank{xy} \le \rank{x} + \rank{y}$; 

	\item[\textbf{R3:}] If $\rank{xy} = \rank{x}$ then $\rank{xyz} = \rank{xz}$;

	\item [\textbf{R4:}] If $\rank{xyz} = \rank{x} + \rank{yz}$ then $\rank{xy} = \rank{x} + \rank{y}$.\footnote{Since set union is commutative, it also follows that if $\rank{xyz} = \rank{x} + \rank{yz}$ then $\rank{xz} = \rank{x} + \rank{z}$.}
\end{description} 
\label{defin:diversity_rank}
\end{definition}
Note that the set $M$ may be infinite, but the subsets $x, y, \ldots$ in the domain of the diversity rank function must be finite. This is intentional: in the case of relational diversity (Section \ref{reldiv}), for example, we can have potentially infinitely many variables, but we are only talking about dependence/independence between finite sets of variables. 

Intuitively, the diversity rank of $x$ is the amount of ``diversity" or ``variation" that $x$ contains. For example, if $x$ is a sequence of vectors in a vector space, the amount of diversity in $x$ is revealed by the dimension of the subspace spanned by $x$. 
If $x$ is the set of attributes in a relation schema in a given database schema, the amount of diversity in $x$ is revealed by the maximum number of different tuples (records) that may exist in the corresponding relation in a given database instance that can be considered as valid for that schema. 
Finally, if $x$ is simply a word in a finite alphabet, a possible measure of the amount of diversity in $x$ is  the number of different letters in $x$, so that for example the diversity of ``abbab'' is $2$ and the diversity of ``abcdda'' is $4$.

It is obvious that we have to require \textbf{R1} and \textbf{R2}. The empty set cannot manifest any diversity, more elements means more diversity, and the amount of diversity manifested by two sets taken together is at most the sum of the amounts of diversity occurring in each of them separately. 

The axioms \textbf{R3} and \textbf{R4} are less intuitive. In brief, Axiom \textbf{R3} states that if adding $y$ to $x$ does not increase the amount of diversity of $x$ (that is, $y$ is ``trivial'' given $x$), then adding it to $xz$ does not increase the amount of diversity of $xz$ (that is, $y$ is also ``trivial'' given $xz$) either; and Axiom \textbf{R4} states that if adding $yz$ to $x$ increases the amount of diversity of the maximum amount possible (that is, $yz$ is ``maximally non-trivial'' given $x$) then adding $y$ to $x$ also increases the amount of diversity of the maximum amount possible (that is, $y$ is also ``maximally non-trivial'' given $x$). 

To better understand the roles of these axioms, let us briefly consider toy examples that would violate them. 

For \textbf{R3} let us consider the function, defined over the two-element set $M = \{a,b\}$, such that $\rank{\emptyset} = \rank{a} = \rank{b} = 0$ but $\rank{ab} = 1$. According to this candidate diversity rank function, $\{a\}$ is ``trivial'' with respect to the empty set (i.e., it adds no further diversity if added to it), but it is not so with respect to the bigger set $\{b\}$: indeed, $\rank{\emptyset a} = \rank{\emptyset} = 0$ but $\rank{\emptyset a b} = 1 > \rank{\emptyset  b} = 0$. The purpose of \textbf{R3} is to prevent this kind of scenario, in which $a$ contributes more diversity in the presence of some other element $b$ than in its absence. 

For \textbf{R4}, instead, let us consider a function over three elements $a$, $b$ and $c$ such that $\rank{a} = \rank{b} = \rank{c} = 1$, $\rank{ab} = 1.5$ , $\rank{ac} = 2$ and $\rank{abc} = 3$. Here $b$ is not maximally diverse with respect to $a$, in the sense that $\rank{ab} < \rank{a} + \rank{b}$; however, $b$ is maximally diverse with respect to $ac$, in the sense that $\rank{abc} = \rank{ac} + \rank{b}$. This type of scenario, in some of the diversity in $b$ would appear to be subsumed by the diversity in $a$ but not by the diverstiy in $ac$, is what we aim to prevent by this rule.

\textbf{R3} and \textbf{R4} (as well as the right part of \textbf{R2}) would follow immediately if we assumed that our diversity rank function is \emph{submodular}, in the sense that it satisfies the condition 
\begin{description}
	\item[\textbf{SUBM:}] $\rank{xyz} + \rank{z} \leq \rank{xz} + \rank{yz}:$ 
\end{description}
\begin{proposition}
Every function from finite subsets of some set $M$ to non-negative real numbers satisfying \textbf{R1}, the left part of \textbf{R2} and \textbf{SUBM} is a diversity rank function in the sense of Definition \ref{defin:diversity_rank}. 
\label{propo:subm_implies}
\end{proposition}
\begin{proof}
Let $\rank{x}$ satisfy \textbf{R1}, the left part of \textbf{R2} and \textbf{SUBM}. We need to show that $\rank{x}$ also satisfies the right part of \textbf{R2} as well as \textbf{R3} and \textbf{R4}. 
Choosing $z = \emptyset$, we obtain immediately from submodularity that $\rank{xy} + \rank{\emptyset} \leq \rank{x} + \rank{y}$. But by \textbf{R1} we know   $\rank{\emptyset} = 0$, and subadditivity (that is, the right part of \textbf{R2}) follows. 
	As for \textbf{R3}, suppose   $\rank{xy} = \rank{x}$. Then by \textbf{SUBM}, $\rank{xyz} \leq \rank{xy} + \rank{xz} - \rank{x} = \rank{x} + \rank{xz} - \rank{x} = \rank{xz}$; but on the other hand $\rank{xz} \leq \rank{xyz}$ by the left part of \textbf{R2} and so $\rank{xyz} = \rank{xz}$ as required. 
Finally, \textbf{R4} holds. Indeed, by \textbf{SUBM} we know   $\rank{xyz} + \rank{y} \leq \rank{xy} + \rank{yz}$. Thus, if $\rank{xyz} = \rank{x} + \rank{yz}$ we have immediately that  $\rank{x} + \rank{yz} + \rank{y} \leq \rank{xy} + \rank{yz}$, that is,  $\rank{x} + \rank{y} \leq \rank{xy}$. But $\rank{xy} \leq \rank{x} + \rank{y}$ by the right part of \textbf{R2}, which we already proved, and hence $\rank{xy} = \rank{x} + \rank{y}$ as required. 
\end{proof}

The converse of the above result, however, is not true: there exist diversity rank functions that do not satisfy submodularity. As a toy example, consider the diversity rank function over the set $\{a,b,c\}$ defined as
\[
\begin{array}{l c l}
\rank{\emptyset} = 0;        &~~~~~~~~~~~~& \rank{ab} = 2.1;\\
\rank{a} = \rank{b} = 1.5; & & \rank{ac} = \rank{bc} = 1.6;\\
\rank{c} = 1; & & \rank{abc} = 3.
\end{array}
\]
Here $\rank{abc} + \rank{c} = 3 + 1 = 4$, but $\rank{ac} + \rank{bc} = 1.6 + 1.6 = 3.2$; therefore, \textbf{SUBM} fails. 


Of course, one needs to verify that this is indeed a diversity rank function. \textbf{R1} holds, because $\rank{\emptyset}$ is indeed $0$; \textbf{R3} holds, because its premise is never satisfied non-trivially (it is never the case that $\rank{xy} = \rank{x}$ for $y \not = \emptyset$); \textbf{R4} likewise holds because its premise is never satisfied non-trivially; the left side of \textbf{R2} is easily verified by inspection; and as for the right part of \textbf{R2} we need to verify the following six cases: 
\begin{enumerate}
    \item $\rank{ab} = 2.1 \leq \rank{a} + \rank{b} = 1.5 + 1.5 = 3$; 
    \item $\rank{ac} = 1.6 \leq \rank{a} + \rank{c} = 1.5 + 1 = 2.5$; 
    \item $\rank{bc} = 1.6 \leq \rank{b} + \rank{c} = 1.5 + 1 = 2.5$; 
    \item $\rank{abc} = 3 \leq \rank{a} + \rank{bc} = 1.5 + 1.6 = 3.1$; 
    \item $\rank{abc} = 3 \leq \rank{b} + \rank{ac} = 1.5 +1.6 = 3.1$; 
    \item $\rank{abc} = 3 \leq \rank{c} + \rank{ab} = 1 + 2.1 = 3.1$. 
\end{enumerate}
(The other possible cases follow from these because of the left side of rule \textbf{R2}, e.g. $\rank{ab} + \rank{ac} \geq \rank{ab} + \rank{c} \geq \rank{abc}$).
 
It could appear tempting at this point of the work to exclude counterexamples such as this one from consideration by adding submodularity to the definition of diversity rank function.  But as Proposition \ref{propo:subm_fails_relational} will show, there exist diversity rank functions of practical interest that fail to satisfy \textbf{SUBM}.\\

A direct consequence of Proposition \ref{propo:subm_implies} is that our diversity rank functions generalize matroids:
\begin{definition}[Matroid]
	A matroid $r$ over some finite set $E$ is a function from subsets of $E$ to non-negative integers satisfying the following conditions for all subsets $x$ and $y$ of $E$.
\begin{description}
	\item[\textbf{M1}] $r(x) \leq \vert x \vert $;\footnote{As is customary, $\vert x\vert$ here defines the cardinality of the set $x$. This is not necessarily the same as $\rank{x}$, which in this work is the diversity rank of $x$ with respect to some diversity rank function $\rank{\cdot}$.}
	\item[\textbf{M2}] $r(xy) + r(x \cap y) \leq r(x) + r(y)$; 
	\item[\textbf{M3}] If $\vert y\vert = 1$ then $r(x) \leq r(xy) \leq r(x) + 1$.
\end{description}
\end{definition}
The above is just one of several equivalent definitions used in the literature. It is equally possible to define a matroid in terms of its \emph{independent sets} (that is, the $x$ such that $r(x) = \vert x\vert$), in terms of \emph{bases} (maximal independent sets), in terms of \emph{circuits} (minimal non-independent sets), or in terms of \emph{closure operations}. All these definitions can be shown to give rise to the same class of mathematical objects: we refer the reader to \cite{oxley2006matroid} for more details. 

\begin{corollary}
Let $r$ be a matroid over some finite set $E$. Then $r$ is a diversity rank function over $E$. 
\label{coro:mat_div}
\end{corollary}
\begin{proof}
\textbf{R1} holds for $r$. Indeed, by \textbf{M1}, $r(\emptyset) \leq \vert \emptyset \vert = 0$, and therefore the only possibility is that $r(\emptyset) = 0$. 
The left part of \textbf{R2} holds immediately because of \textbf{M3} (this can be shown by easy induction on the the number of elements in $y \backslash x$).
Moreover, \textbf{SUBM} also holds of $r$. Indeed, since $xyz = xzyz$, we have by \textbf{M2},  $r(xyz) + r(xz \cap yz) \leq r(xz) + r(yz)$. But $z \subseteq xz \cap yz$, and so by the left part of \textbf{R2} (which we already proved), $r(z) \leq r(xz \cap yz)$ and hence $r(xyz) + r(z) \leq r(xz) + r(yz)$, as required. 
The conclusion then follows, since because of Proposition \ref{propo:subm_implies} the right part of \textbf{R2} as well as \textbf{R3} and \textbf{R4} are also true of $r$. 
\end{proof}

Given a notion of diversity, it is easy to define dependence and independence in terms of \emph{minimal} and \emph{maximal} diversity contributions: 

\begin{definition}\label{dep}
Suppose $M$ is a set and $\rank{\cdot}$ a diversity rank function on $M$. We can now define {\em dependence} relations between finite subsets of $M$ (with respect to $\rank{\cdot}$) as follows: \begin{itemize}

\item {\bf Dependence:} $y$ {\em is totally determined by} (or {\em depends on}) $x$, in symbols $\fdep(x,y)$, if and only if $\rank{xy}=\rank{x}$.
\item {\bf Constancy:} $x$ {\em is constant},  in symbols $\fdep(x)$, if and only if $\rank{x}=0$.
\item {\bf Independence:} $x$ and  $y$ are {\em independent}, in symbols $x \boto y$, if and only if $\rank{x}+\rank{y}=\rank{xy}$.
\end{itemize} 
\label{defin:depindep}
\end{definition}

The idea is that $\fdep(x,y)$ holds under a diversity rank function if the amount of diversity
inherent in $x$ in terms of the rank function  does not increase when $y$ is added. Simply put, $x$ determines $y$, so no new diversity occurs. $\fdep(x)$, on the other hand, holds if $x$ has no diversity at all; and  $x\boto y$ holds if the diversity inherent in $x$ is so unrelated to the diversity inherent in $y$ that when the two are put together into $xy$, the diversity is the sum of the diversity of $x$ and the diversity of $y$: no loss of diversity occurs because there is---intuitively---no connection between $x$ and $y$.


%

\section{Examples}
Let us now consider some examples of our definitions, in order to get a better feel of their applicability and consequences. As we will see, aside from the example of Subsection \ref{reldiv} (relational diversity), all our examples also satisfy the submodularity condition. Nonetheless, it is worth emphasizing that, when the ranks are not necessarily integer, these diversity rank functions are not matroids. 

There are various diversity ranks, also called diversity measures or diversity indexes,  in biology, such as Shannon index, Simpson index, Renyi index, $\alpha$-diversity, $\beta$-diversity or $\gamma$-diversity,  each attempting to describe in terms of a single real number   the number of different species in an area, the number of individuals of each species, and other parameters relevant in characterizing diversity of a biological system \cite{10.2307/2333344,Hill,magurran2004measuring}.
Similar diversity ranks have been introduced in other areas, as well, and usually take the concept of entropy (see subsection 3.6 below) as a starting point.

\subsection{Constant diversity}
One extreme case is the constant rank for which $\rank{\emptyset}=0$ and, for some number $c$,  $\rank{x}=c$ for all $x\subseteq_f M$ with $x\ne\emptyset$.
If $c=0$,  there is no diversity: every set depends on every other set and is also independent of every other set. If $c\ne 0$, then every set $y$ is still dependent on any non-empty set $x$, because $\rank{xy} = c = \rank{x}$, and every set $x$ is still independent from the empty set $\emptyset$, because $\rank{x\emptyset} = \rank{x} = \rank{x} + \rank{\emptyset}$; but two non-empty sets $x$ and $y$ are not independent, because $\rank{xy} = c \not = c + c = \rank{x} + \rank{y}$. 

Constant diversity is trivially submodular: if $z$ is empty then $\rank{xyz} + \rank{z} = 2c = \rank{xz} + \rank{yz}$, and if instead $z$ is empty then submodularity reduces to $\rank{xy} \leq \rank{x} + \rank{y}$ (which is part of rule \textbf{R2}, and which is also easily verified for this diversity function). 

\subsection{Singular diversity}
Let $a_0\in M$ be fixed. Let
$$\rank{x}=\left\{\begin{array}{ll}
1,&\mbox{if $a_0\in x$,}\\
0,&\mbox{otherwise.}
\end{array}\right.$$ In this case a selected element $a_0$ is the only source of ``diversity" there is. A set has diversity $1$ if and only if it contains $a_0$. In this case $y$ depends on  $x$ if $$a_0\in y\to a_0\in x$$
and two sets $x$ and $y$ are independent if at most one of them contains $a_0$. 
So dependence reduces in this case to implication and independence to the Sheffer stroke (also known as NAND). 

Again, this diversity rank function is submodular: if $a_0 \in z$ then $\rank{xyz} + \rank{z} = 2 = \rank{xz} + \rank{yz}$, and if $a_0 \not \in z$ then $\rank{xyz} + \rank{z} = \rank{xy} \leq \rank{x} + \rank{y} = \rank{xz} + \rank{yz}$. 

Note that the notion of dependence arising from singular diversity is not symmetric: in particular, the empty set depends on $\{a_0\}$ but $\{a_0\}$ does not depend on the empty set. 

\subsection{Two-valued diversity}
Suppose $\rank{\{a\}}$ is either $0$ or $1$ for all $a\in M$. Then, by Rule \textbf{R2}, it follows that  
$$\rank{x}=\max\{\rank{\{a\}} : a \in x\}$$ 
	for all $x \subseteq_f M$. Hence, it suffices in this case to declare which elements
have diversity 1; a set has diversity 1 if and only one if it contains one
of those elements.\footnote{Hence, a singular diversity is a special case of a two-valued diversity where only one element is
declared to have diversity 1.}

	In this case diversity is an on/off phenomenon, either it exists (1) or it does not (0), and a set has diversity if it includes some singleton that has it.
In an extreme case $\rank{\{a\}}=0$ for {\em all} singletons $\{a\}$, $a\in M$, and we have zero constant diversity: every set has diversity 0. In another extreme case $\rank{\{a\}}=1$ for {\em all} singletons $\{a\}$, $a\in M$, and we are again in constant diversity: every non-empty set has diversity 1.

According to this diversity notion, $=\!\!(x, y)$ if and only if $\exists a \in x \text{ such that } \rank{\{a\}} = 1 \Rightarrow \exists b \in y \text{ such that } \rank{\{b\}} = 1$; and $x$ is independent from $y$ if and only if at most one of $x$ and $y$ contain an element $c$ with $\rank{\{c\}} = 1$.

Submodularity is, once more, straightforwardly verified. $\rank{xyz} + \rank{z}$ is $2$ if and only if $\rank{z} = 1$, and in that case it is easy to see that $\rank{xz} + \rank{yz}$ is also $2$. $\rank{xyz} + \rank{z}$ is $1$ if $\rank{xy} = 1$, in which case at least one of $\rank{x}$ and $\rank{y}$ is $1$ and hence $\rank{xz} + \rank{yz} \geq 1$; and finally, if $\rank{xyz} + \rank{z} = 0$ then it is trivially true that $\rank{xyz} + \rank{z} \geq \rank{xz} + \rank{yz}$. 

\subsection{Uniform  diversity}
Suppose
$$\rank{x}=\vert x\vert.$$ This is the  choice of taking the cardinality of the (finite) set $x \subseteq_f M$ as the measure of its diversity. 
Dependence means inclusion: $y$ is totally determined by  $x$ if and only if $\vert xy\vert = \vert x\vert$, that is, if and only if $y\subseteq x$. Independence is disjointness: $x$ 
and $y$ are independent if and only if $\vert xy\vert = \vert x\vert + \vert y\vert $, that is, if and only if $x\cap y=\emptyset$. Once more, submodularity is easily verified by observing that $\vert xyz\vert  = \vert xz\vert  + \vert yz\vert - \vert (xz) \cap (yz)\vert$ and that $\vert (xz) \cap (yz)\vert > \vert z\vert$. \\

\noindent \textbf{Remark:} Uniform diversity shows that if $M$ has at least three elements $a,b,c$  then independence is not necessarily equivalent to the failure of dependence both ways. Indeed, $\{a,b\}$ is not dependent on $\{b,c\}$ or vice versa, since neither set is contained in the other, but $\{a,b\}$ and $\{b,c\}$ are not independent either since they are not disjoint. This is not surprising: in our framework, $y$ depends on $x$ if adding $y$ to $x$ contributes no diversity whatsoever to it, while $x$ and $y$ are independent if adding $y$ to $x$ contributes the maximal amount $\rank{y}$ of diversity to it. But it is certainly possible for neither extreme to be the case.

\subsection{Coverage diversity}
 Suppose $U$ is a finite set and choose $A_a\subseteq U$ for each $a\in M$. For $a_1 \ldots a_n \in M$, let 
$$\rank{\{a_1,\ldots, a_n\}}=\vert A_{a_1}\cup\ldots\cup A_{a_n}\vert.$$ We can think of each $A_a$ as  ``data", about the element $a$ of $M$. The more data we have the more diversity we give to the element, and the diversity of a set is obtained by simply putting together all the data we have. In this simple example the data is not thought to be specific to the elements of $M$, so the data about different elements is just lumped together.  For example, if $a$ and $b$ are two botanic genera in the bean family, the diversity of $\{a,b\}$ in a set $U$ of data about species (e.g. in some location) is obtained by counting how many different species of the two genera there are in $U$. 

	According to this diversity notion, $y$ is dependent on $x$ if and only if $\bigcup \{A_{a} : a \in y\} \subseteq \bigcup \{A_{b} : b \in x\}$, that is, every data point corresponding to some element of $y$ also corresponds to some element of $x$; and $y$ is independent from $x$ if and only if  $\bigcup \{A_{a} : a \in y\} \cap \bigcup \{A_{b}: b \in x\} = \emptyset$, that is, if no data point corresponds to some element of $x$ \emph{and} to some element of $y$.\\
	
	Submodularity can be verified much as in the previous example: indeed, $\rank{xyz} = \vert \bigcup A_a : a \in x \cup y \cup z\vert = \rank{xy} + \rank{xz} - \vert C\vert$ for $C = \{u: u \in A_b \cap A_c \text{ for some } b \in x \cup z, c \in y \cup z\}$, and it is easy to see that $\vert C\vert \geq \rank{z}$. 

\subsection{Entropy}
\label{entropy_diversity}
Let us think of the elements of $M$ as \emph{discrete random variables} $v_1, v_2, \ldots$ over some probability space and with outcomes in some finite set $A$.\footnote{Nothing in this example hinges on $A$ being the same for all $v \in M$, but we will assume so for simplicity.} Then for any $x = \{v_1 \ldots v_k\} \subseteq_f M$ we can define $\rank{x}$ as the joint entropy \cite{borda2011fundamentals} $H(x)$ of $v_1 \ldots v_k$, that is, as\footnote{In this work, $\log$ will always represent the base-$2$ logarithm.}  
\[
 - \sum_{(m_1 \ldots m_k) \in A^k} P(v_1 \ldots v_k = m_1 \ldots m_k) \log P(v_1 \ldots v_k = m_1 \ldots m_k). 
\]

This definition clearly satisfies rule \textbf{R0}, since the entropy of the only possible distribution over the empty space is zero; moreover, it is not hard to convince oneself that it is monotone and submodular. In brief, this can be shown by considering the \emph{conditional entropy} $H(y \mid x) = H(xy) - H(x)$. 

	Indeed, it can be proved (see any Information Theory textbook, for instance Theorem 2.2.1 of \cite{cover1991entropy}) that the conditional entropy $H(y \mid x)$ is always non-negative\footnote{More precisely, this theorem shows that $H(xy)-H(x) = -\sum_{m} P(x = m) \sum_{m'} P(y=m' \mid x=m)\log P(y=m'\mid x=m)$, and the right hand side is straightforwardly seen to be non-negative.}, from which we have the left part of \textbf{R2};  furthermore (see e.g. Theorem 2.6.5 of \cite{cover1991entropy})\footnote{Strictly speaking, this theorem states that $H(x) - H(x \mid y) \geq 0$, but if we consider the above inequality with respect to distributions already conditioned on $z$ the result follows.}  $H(x \mid yz) \leq H(x \mid z)$, from which we obtain immediately   $H(xyz) - H(yz) \leq H(xz) - H(z)$, that is, Axiom \textbf{SUBM}. 

From Proposition \ref{propo:subm_implies}, we can immediately conclude that entropy is an example of a (submodular) diversity rank function. Here,   
$y$ depends on $x$ according to the entropy diversity rank if and only if $H(xy) = H(x)$, that is, if and only if the relative entropy of $y$ given $x$ is $0$, or in other words if the value of $y$ is completely determined by the value of $x$; and $x$ and $y$ are independent according to this rank if and only if they are independent sets of random variables, that is, $P(x = m, y=m') = P(x=m)P(y=m')$ for all possible choices of values $a$ and $b$ for $x$ and $y$.  We observe that this notion of probabilistic independence of random variables is exactly the one axiomatized by Geiger, Paz and Pearl in \cite{MR1097266}. Their axiomatization will be the base for our axiomatization of independence in our more general setting in Section \ref{div_to_ind}. 

\subsection{Relational diversity} \label{reldiv} Suppose $X$ is a nonempty, finite set of \emph{variable assignments} $s$ from a set $V$ of variables to a set $A$ of elements (in the language of Dependence and Independence Logic, $X$ is said to be a \emph{team} over $A$ with domain $V$; and it is not hard to see that it is equivalent to a relational table in which every variable indicates a different column).\footnote{In general, in Dependence and Independence Logic teams do not necessarily have to be finite, but we will focus on the finite case in this example.} Given some $x = \{v_1 \ldots v_n\} \subseteq_f V$, let
$$\rank{x}=
\log(\rows_X(v_1 \ldots v_n)).$$

where $$\rows_X(v_1 \ldots v_n)=\vert \{(s(v_1),\ldots,s(v_n)):s\in X\}\vert $$ is the number of different values that $x =v_1 \ldots v_n$ takes in $X$.\footnote{It is easy to check that this value does not depend on the ordering of $v_1 \ldots v_n$.}

We can think of each $s\in X$ as an ``observation", or ``data", about the possible values that the variables in $V$ can take. The more different observations we have the more diversity we give to the element. Note the difference with  coverage diversity, where the data was not specific to the element of $A$. Here what matters is the relationships of the different observations to each other. Thus $$\rank{\{v\}}=\log\vert \{s(v):s\in X\}\vert,$$ that is, the diversity rank of a single element $v$ of $V$ is the (logarithm of the) number of different observations about $v$. The diversity of a pair $\{v,w\}$ is the (logarithm of the) number of different combinations of observations of $v$ and $w$. For example, if $v$ and $w$ are two genera, the diversity of $\{v,w\}$ in a set $X$ of observations is calculated by counting how many different pairs of observations of a specimen of $v$ and a specimen of $w$ there are in $X$. 

The presence of the logarithm operator in this definition may appear at first sight somewhat outlandish, but it is in fact quite natural. Indeed, one may recall that, in information theory, the \emph{information content} of an event is defined as the negative logarithm of its probability, and the information content of a random variable (i.e. its entropy) is the expectation of the information content of it taking all possible values 
\cite{borda2011fundamentals,shannon1948mathematical}. Therefore, if we associate a relation $R$ with the probability distribution selecting any tuple in $R$ with equal probability $1/\vert R\vert$, the information content of this distribution (and, hence, of the relation) is precisely $-\log(1/\vert R \vert) = \log(\vert R\vert)$.

The dependence relation arising from the relational diversity rank is the usual functional dependence relation of database theory and dependence logic \cite{armstrong,codd1970relational,MR2351449}. Why?
 By definition, $=\!\!(x,y)$ if and only if $\log(\rows_X(xy)) = \log(\rows_X(x))$, that is, if and only if $\rows_X(xy) = \rows_X(x)$.  This can be the case if and only if any two $s, s' \in X$ which differ with respect to $xy$ differ already on $x$ alone, or, by contraposition, if and only if any two $s, s' \in X$ which are the same with respect to $x$ are also the same with respect to $y$. This is precisely the usual notion of database-theoretic functional dependence \cite{armstrong}, which is arguably the most important (although by no means the only) notion of dependence studied in the context of database theory. 
 
 The independence relation arising from the relational diversity rank is also the independence relation of Independence Logic \cite{GV}.  Indeed, by definition, $\indepc{x}{y}$ if and only if $\rows_X(xy) = \rows_X(x) \cdot \rows_X(y)$. Then, enumerating the elements of $x$ as $(v_1 \ldots v_k)$ and the elements of $y$ as $(w_1 \ldots w_{t})$, it is always the case that $\rows_X(xy) = \vert\{(s(v_1) \ldots s(v_k), s(w_1) \ldots s(w_{t})) : s \in X\}\vert \leq \vert\{(s(v_1) \ldots s(v_k)) : s \in X\}\vert \cdot \vert\{(s'(w_1) \ldots s'(w_{t})) : s' \in X\}\vert = \rows_X(x) \cdot \rows_X(y)$, since every possible value for $(v_1 \ldots v_k, w_1 \ldots w_t)$ in $X$ corresponds to one possible value for $(v_1 \ldots v_k)$ and one possible value for $(w_1 \ldots w_k)$ in $X$. Equality holds if and only if the converse also holds, i.e., if and only if for any two $s, s' \in X$ there exists some $s'' \in X$ such that $s''(v_1 \ldots v_k) = s(v_1 \ldots v_k)$ and $s''(w_1 \ldots w_{t}) = s'(w_1 \ldots w_{t})$. In particular, this implies at once that if $x \boto y$ the variables occurring in both $x$ and $y$ must take only one value in $X$, and that if $x$ and $y$ are disjoint and $x \boto y$ then the projection of $X$ onto $xy$ must be the Cartesian product of its projections onto $x$ and $y$.\footnote{Note that this is not the same as saying that the projection of $X$ onto $xy$ is the \emph{Natural Join} of the projections of $X$ onto $x$ and onto $y$. For example, consider the relation $X = \{(0,1,0),(1,2,1)\}$ over three variables named $v_1, v_2, v_3$ respectively.
Then the natural join of the projections of $X$ onto $v_1 v_3$ and $v_2 v_3$ is its projection onto $v_1 v_2 v_3$; however, $v_1 v_3 \boto v_2 v_3$ is not the case (indeed, $\rows(v_1 v_3) \cdot \rows(v_2 v_3) = 2 \cdot 2 = 4$ but $\rows(v_1 v_2 v_3) = 2$).
}  
 

It may be instructive to verify that the relational diversity notion of rank satisfies our axioms:
\begin{description}
	\item[\textbf{R1:}] Since $\rows(\emptyset) = \vert\{()\}\vert = 1$ for any choice of $X$, where $()$ represents the empty tuple, we have $\rank{\emptyset} = 0$, as required.

	\item[\textbf{R2:}] Since $\rows(x) \leq \rows(xy)$ and the logarithm is a monotone function, we have immediately that $\rank{x} \leq \rank{xy}$; and since $\rows(xy) \leq \rows(x) \cdot \rows(y)$, we have immediately that $\rank{xy} \leq \rank{x} + \rank{y}$.
	\item[\textbf{R3:} ] If $\rank{xy} = \rank{x}$, $\rows(xy) = \rows(x)$ and hence every possible value of $x$ occurs together with only one possible value of $y$. But then every possible value of $xz$ occurs together with only one possible value of $y$, and hence $\rows(xyz) = \rows(xz)$ and $\rank{xyz} = \rank{xz}$; 
	\item[\textbf{R4:}] If $\rank{xyz} = \rank{x} + \rank{yz}$, it must be the case that $\rows(xyz) = \rows(x) \cdot \rows(yz)$, and hence that every possible value of $x$ occurs together with every possible value for $yz$. But then in particular every possible value for $x$ occurs together with every possible value for $y$, and so $\rows(xy) = \rows(x) \cdot \rows(y)$ and $\rank{xy} = \rank{x} + \rank{y}$. 
\end{description}

In contrast to our other examples, this notion of relational diversity is \emph{not} submodular, as the following counterexample, which we owe to Tong Wang\footnote{Personal communication.} shows:  

\begin{proposition}
Relational diversity fails to satisfy \textbf{SUBM}. 
\label{propo:subm_fails_relational}
\end{proposition}
\begin{proof}
Consider the relation $X$
\begin{center}
$$\begin{array}{ccc}
v_1& v_2& v_3\\
\hline
1 &1& 1\\
1 &1& 2\\
2 &1& 1\\
1 &2& 1\\
2 &1& 2\\
\end{array}$$
\end{center}
	Then $\rows_X(v_1v_2v_3) = 5$, $\rows_X(v_2) = 2$, and $\rows_X(v_1v_2) = \rows_X(v_2v_3) = 3$. Thus, $\rows_X(v_1v_2v_3) \cdot \rows_X(v_2) = 10 > 9 = \rows_X(v_1v_2) \cdot \rows_X(v_2v_3)$, and hence $\rank{v_1v_2v_3} + \rank{v_2} > \rank{v_1v_2} + \rank{v_2v_3}$. 
\end{proof}

\subsection{Algebraic diversity} Suppose that $V$ is a vector space and that $h$ maps $M$ into $V$. We obtain a diversity rank function by letting for finite $x\subseteq M$:\footnote{We refer the reader to any algebra textbook, for example to \cite{lang2002algebra}, for the relevant algebra background.}

$$\rank{x}=\mbox{ the dimension of the subspace generated by }\{h(a):a\in x\}.$$
Submodularity \textbf{SUBM} follows from the well known fact that if $U$ and $V$ are vector subspaces,
\[
	\dim(U \cup V) = \dim(U) + \dim(V) - \dim(U \cap V).
\]

In this context it is not hard to verify that $V$ is dependent on $U$ if and only if 
$\dim(U \cup V) = \dim(U)$, that is, if and only if every vector of $V$ is a linear combination of vectors in $U$; and that, on the other hand, $U$ and $V$ are independent if and only if $\dim(U \cup V) = \dim(U) + \dim(V)$, that is, if and only if the subspace generated by $U$ and the subspace generated by $V$ do not
share a nonzero vector.

Likewise if $F$ is a field, we get a diversity rank function by letting for finite $x\subseteq M$ and letting $h$ map $M$ into $F$ instead:
$$\rank{x}=\mbox{ the transcendence degree of the subfield generated by }\{h(a):a\in x\}.$$
This gives rise to the concepts of algebraic dependence and independence.

As mentioned in the Introduction, this notion of rank defines a matroid (in fact, it was one of the original motivations for the development of Matroid Theory); and thus, by Corollary \ref{coro:mat_div}, it is also a diversity rank function according to our definition.

%
%

\section{From diversity to dependence}

Given a diversity function $\rank{\cdot}$, we have defined {\em dependence}  $\fdep(x,y)$ of $y$ on $x$ by $\rank{xy} = \rank{x}$. 

It is easy to verify that dependence satisfies the following axioms:
\begin{proposition}\label{aa} Dependence satisfies the following properties:
\begin{description}
\item[1. Reflexivity:]  $\fdep(xy, x)$.
\item[2. Augmentation:] $\fdep(x,y)$ implies $\fdep(xz,yz)$. 
\item[3. Transitivity:] If $\fdep(x,y)$ and $\fdep(y,z)$, then $\fdep(x,z)$.
\end{description}
\label{propo:ax_dep}
\end{proposition}
\begin{proof}
\begin{description}
\item[Reflexivity:] Clearly $\rank{xyx}=\rank{xy}$. Therefore, $=\!\!(x y,  x)$.
\item[Augmentation:] Suppose   $\rank{ x  y} = \rank{ x}$. Then, by \textbf{R3}, $\rank{ x y  z} = \rank{ x  z}$; and therefore, $\rank{ x  z  y  z} = \rank{ x  z}$, or, in other words, $=\!\!( x  z,  y  z)$. 
\item[Transitivity:] Suppose   $\rank{ x} = \rank{ x  y}$ and $\rank{ y} = \rank{ y  z}$. Again, by \textbf{R3}, from $\rank{ x} = \rank{ x  y}$ we get   $\rank{ x  z} = \rank{ x  y  z}$; and similarly, from $\rank{ y} = \rank{ y  z}$ we get   $\rank{ x  y} = \rank{ x  y  z}$. By the transitivity of equality, we can conclude   $\rank{ x  z} = \rank{ x  y}$. But we have as an hypothesis   $\rank{ x  y} = \rank{ x}$, and therefore we can conclude   $\rank{ x} = \rank{ x  z}$, or, in other words,   $\fdep(x,  z)$.
\end{description}
\end{proof}

We can use the above rules as the axioms of a proof system for inferring the consequences of a set of dependence assertions. More precisely, given a set $\Sigma$ of dependence assertions and a dependence
assertion $=\!\!(x, y)$ for $x, y \subseteq_f M$, we will write that $\Sigma\  \vdash_D\  =\!\!(x, y)$ if it is possible to derive $=\!\!(x,y)$ from $\Sigma$ through applications of the rules of Reflexivity, Augmentation and Transitivity.

It follows from these axioms that a dependency notion is entirely defined even if we only consider singletons on the right-hand side of it: 

\begin{proposition}
Let $\Sigma$ be a  set of assertions of the form $=\!\!(z, w)$ for $z, w \subseteq_f M$, and let also $x, y \subseteq_f M$. Then $\Sigma \  \vdash_D\  =\!\!(x, y)$ if and only if $\Sigma \  \vdash_D\  =\!\!(x, \{a\})$ for all $a \in y$. 
\label{prop:depsing}
\end{proposition}
\begin{proof}
By Reflexivity, if $a \in y$ then it is always the case that $\Sigma \  \vdash_D\  =\!\!(y, \{a\})$. If $\Sigma \  \vdash_D\  =\!\!(x, y)$, by Transitivity it is thus the case that $\Sigma \  \vdash_D\  =\!\!(x, \{a\})$ for all such $a$. 
Conversely, suppose   $\Sigma \  \vdash_D\  =\!\!(x, \{a\})$ for all $a \in y$. Then, in order to reach our conclusion that $\Sigma \  \vdash_D\  =\!\!(x, y)$, it suffices to verify that whenever $\Sigma \  \vdash_D\  =\!\!(x, y_1)$ and $\Sigma \  \vdash_D\  =\!\!(x, y_2)$ it is also the case that $\Sigma \  \vdash_D\  =\!\!(x, y_1 y_2)$. 
This is easily shown: if $\Sigma \  \vdash_D\  =\!\!(x, y_1)$, by Augmentation we have 
 $\Sigma \  \vdash_D\  =\!\!(x, x y_1)$ (remember that in our notation $xx = x\cup x = x$), and if $\Sigma \  \vdash_D\  =\!\!(x, y_2)$ again by Augmentation we have $\Sigma \  \vdash_D\  =\!\!(x y_1, y_1 y_2)$, and an application of Transitivity gives us $\Sigma \  \vdash_D\  =\!\!(x, y_1 y_2)$. The conclusion follows immediately. 
\end{proof}

The following is essentially proved in \cite{armstrong}, albeit in the special case of relations and functional dependencies:

\begin{theorem}[Completeness of the Dependence Axioms]\label{armstr} Let $\Sigma$ be a  set  of assertions of the form $\fdep(z,w)$, where all $z$ and $w$ are finite subsets of some set $M$ and let also $x, y \subseteq_f M$. The following properties are equivalent:
\begin{enumerate}

\item  $\fdep(x,y)$ holds under {\em every} diversity rank function on $M$ under which $\Sigma$ holds.
\item  $\fdep(x,y)$ holds under {\em every} two-valued diversity rank function  $\P(M)\to \{0,1\}$ under which $\Sigma$ holds. 
\item  $\fdep(x,y)$ holds under {\em every} relational diversity rank function under which $\Sigma$ holds. 
\item $\fdep(x,y)$
 follows from $\Sigma$ by the rules of Proposition~\ref{aa}.
\end{enumerate}\end{theorem}

\begin{proof} Trivially, (1) implies (2) and (3). Furthermore, (4) implies (1) by Proposition \ref{propo:ax_dep}. We demonstrate that (2) implies (4) and (3) implies (4). Let us first assume (2). Suppose $\fdep(x,y)$ does not follow from $\Sigma$ by the rules of Proposition~\ref{aa}.
Let $V$ be the set of $a\in M$ such that  $\fdep(x,\{a\})$ follows from $\Sigma$ by these rules. By Proposition \ref{prop:depsing}, for all $w \subseteq M$, we have that  $\Sigma \  \vdash_D\  =\!\!(x,w)$ if and only if $w \subseteq V$. 

Let $W$ be all the remaining elements of $M$. Since $y \not \subseteq V$, $W\ne\emptyset$. Let us define a diversity rank function on $M$ by letting for $a\in M$: $$\rank{\{a\}}=\left\{\begin{array}{l}
0,\mbox{ if $a\in V$}\\
1,\mbox{ if $a\in W$,}
\end{array}\right.$$ and otherwise $$\rank{\{a_1,\ldots, a_n\}}=\max\{\rank{\{a_1\}},\ldots,\rank{\{a_n\}}\}.$$
Note that $\rank{xy}= 1$, while $\rank{x}=0$. Thus the relation $\fdep(x,y)$ does not hold under this diversity rank function. Suppose then $\fdep(z,w)\in\Sigma$. If $z\subseteq V$, this means   $\Sigma \  \vdash_D\  =\!\!(x, z)$; and then, by Transitivity, $\Sigma \  \vdash_D\  =\!\!(x, w)$ and so $w \subseteq V$ as well. So $\rank{zw}=\rank{z}=0$ and $\fdep(z,w)$ holds. On the other hand, if $z\not\subseteq V$, then $\rank{z}=1$. So $\rank{zw}=\rank{z}=1$, whence $\fdep(z,w)$ holds again.

Let us then assume (3). We proceed as above. Let $X$ consist of the two functions $\{s_1,s_2\}$, where $s_1(a)=0$ for all $a\in M$, $s_2(a)=0$ for $a\in V$ and $s_2(a)=1$ for $a\in W$. We obtain the same rank as above, so we are done. 
\end{proof}

Since -- as we saw -- functional dependence is exactly the dependency notion generated by the relational diversity rank function, we obtain:

\begin{corollary}[Armstrong] A functional dependence follows semantically, for all relations, 
from a given set of functional dependencies if and only if it follows by the rules of Proposition~\ref{aa}.

\end{corollary}

Theorem~\ref{armstr} shows that Armstrong's completeness theorem for functional dependence is actually a more general completeness theorem of dependence relations arising from diversity ranks.

\section{From diversity to independence}
\label{div_to_ind}

We shall now study the properties of the notions of independence arising from our diversity ranks. Let us recall that, according to our definition, $x$ and $y$ are independent ($x \boto y$) if and only if $\rank{xy} = \rank{x} + \rank{y}$. 

Let us begin by observing that, by our definition, $x \boto x$ if and only if $\rank{x} = \rank{x} + \rank{x}$, that is, if and only if $\rank{x} = 0$: $x$ is independent of itself if and only if $x$ contains no diversity whatsoever according to our diversity rank function, that is, if and only if $x$ is \emph{constant} in the sense of Definition \ref{dep}. 

Note that in the probabilistic case, a random variable $v$ is independent of itself, in the sense that $P(v=a \text{ and } v = b) = P(v=a)P(v=b)$ for all possible values $a, b$ of $X$, if and only if $v$ may take only one value with probability $1$, and hence it is constant in the ordinary sense of the word; and similarly, in the logical/relational case a variable $v$ is independent on itself if and only if it takes only possible value for all variable assignments in the set being considered. Thus, our use of the term ``constancy'' is not unmotivated. 

\begin{proposition}\label{gpp} Independence satisfies the following properties:
\begin{description}
\item[1. Empty Set:] $x\boto \emptyset$.

\item[2. Symmetry:] If $x\boto y$, then $y\boto x$.
\item[3. Decomposition:] If $x\boto yz$, then $x\boto y$.\footnote{By the symmetry of union, the Decomposition rule implies that if $x \boto yz$ then $x \boto z$ as well.}

\item[4. Mixing:] If $x\boto y$ and $xy\boto z$, then $x\boto yz$.

\item[5. Constancy:] If $z \boto z$ then $z \boto x$.\footnote{If one is uninterested in independence assertions $x \boto y$ in which $x$ and $y$ overlap, this axiom can be removed. Our proof of Theorem \ref{thm:ind_comp} then reduces essentially to the proof in \cite{MR1097266}.}
\end{description} 
\end{proposition}

\begin{proof}
Let us prove that these axioms follow from our notion of independence: 
\begin{description}
\item[Empty Set:] Since $\rank{\emptyset} = 0$, $\rank{ x} + \rank{\emptyset} = \rank{ x} + 0 = \rank{ x} = \rank{ x\emptyset}$.
\item[Symmetry:] This follows easily from the commutativity of sum and union. If $\rank{ x} + \rank{ y} = \rank{ x  y}$ then $\rank{ y} + \rank{ x} = \rank{ x  y} = \rank{ y  x}$.
\item[Decomposition:] Suppose   $\indepc{ x}{ y  z}$, that is, $\rank{ x} + \rank{ y  z} = \rank{ x  y  z}$. 

By \textbf{R4}, we then have   $\rank{ x  y} = \rank{ x} + \rank{ y}$ and $\indepc{ x}{ y}$.
\item[Mixing:] Suppose   $\rank{xy} = \rank{x} + \rank{y}$ and $\rank{xyz} = \rank{xy} + \rank{z}$. We need to prove  $\rank{xyz} = \rank{x} + \rank{yz}$. 
We begin by observing that $\rank{x} + \rank{y} + \rank{z} = \rank{xy} + \rank{z} = \rank{xyz}$. But by \textbf{R2} $\rank{yz} \leq \rank{y} + \rank{z}$, and therefore $\rank{x} + \rank{yz} \leq \rank{x} + \rank{y} + \rank{z} = \rank{xyz}$. 
On the other hand, again by \textbf{R2}, $\rank{xyz} \leq \rank{x} + \rank{yz}$, and so in conclusion $\rank{xyz} = \rank{x} + \rank{yz}$, as required. 

\item[Constancy:] If $z \boto z$ then $\rank{z} = \rank{z} + \rank{z}$, and hence $\rank{z} = 0$. But then by \textbf{R2} $\rank{x} \leq \rank{xz} \leq \rank{x} + \rank{z} = \rank{x}$, and thus $\rank{xz} = \rank{x} + \rank{z}$ and $z \boto x$. 



%
\end{description}
\end{proof}

Given a set $\Sigma$ of independence assertions and an independence assertion $x \boto y$ for $x, y \subseteq_f M$ , we will write that $\Sigma \vdash_I x \boto y$ if it is possible to derive $x \boto y$ from $\Sigma$ through applications of the rules of Empty Set, Symmetry, Decomposition, Mixing and Constancy.

The following derived rule will be useful:
\begin{proposition}[Constancy Augmentation]
    Given a set $M$ , let $\Sigma$ be a set of
independence assertions over $M$ of the form $z \boto w$ for $z, w \subseteq_f M$, and suppose   $\Sigma \vdash_I u \boto u$ and $\Sigma \vdash_I x \boto y$. Then $\Sigma \vdash_I xu \boto y$
\end{proposition}
\begin{proof}
By Constancy, if $\Sigma \vdash_I u \boto u$ then $\Sigma \vdash_I u \boto xy$, and so by Symmetry $\Sigma \vdash_I xy \boto u$. If furthermore $\Sigma \vdash_I x \boto y$, by Symmetry $\Sigma \vdash_I y \boto x$; and thus, by Mixing, $\Sigma \vdash_I y \boto xu$, and by Symmetry once more $\Sigma \vdash_I xu \boto y$ as required. 
\end{proof}

In \cite{MR1097266}, a sound and complete axiomatization for independence of tuples of \emph{random variables} (as derived from the definition of entropy in Section  \ref{entropy_diversity}) was found, with the additional requirement that the left- and right-hand sides of the independence assertion are disjoint. 

Theorem \ref{thm:ind_comp} below is a generalization of that result to the case of general diversity rank functions, and without that additional requirement. 

First, we will show that that the axioms given above are complete for assertions of the form $\{a\} \boto \{a\}$:

\begin{lemma}[Completeness of Independence Axioms wrt Constancy Assertions]
Given a set $M$, let $\Sigma$ be a set of independence assertions over $M$, and
let $a \in M$. Then the following properties are equivalent:\begin{enumerate}

\item  $\{a\} \boto \{a\}$ holds under {\em every} diversity rank function on $M$ under which $\Sigma$ holds.
\item  $\{a\} \boto \{a\}$ holds under {\em every} relational diversity rank function   under which $\Sigma$ holds. 
\item $\{a\}\boto \{a\}$
 follows from $\Sigma$ by the rules of Proposition~\ref{gpp}.
\end{enumerate}
\label{lemma:comp_ind_const}
\end{lemma}
\begin{proof}
Trivially (1) implies (2) and (3) implies (1). Let us verify that (2) implies (3). Suppose that $\{a\} \boto \{a\}$ does not follow from $\Sigma$ by the above rules. Then let $V$ contain all $b \in M$ such that $\Sigma \vdash_I \{b\} \boto \{b\}$ and let $S$ be a team with domain $M$ over $\{0,1\}$ (that is, a set of functions from $M$ to $\{0,1\}$) that contains all $s: M \rightarrow \{0,1\}$ such that $s(b) = 0$ for all $b \in V$. 

Now let $\rank{\cdot} = \log(\rows_S(\cdot))$ be the relational diversity rank function induced by $S$: as already discussed, such a diversity rank function satisfies an independence assertion $z \boto w$ if and only if any possible values of $z$ and $w$ in $S$ may occur together, or, in other words, if and only if for all $s, s' \in S$ there exists some $s'' \in S$ that agrees with $s$ on $z$ and with $s'$ on $w$. In particular, for the $S$ given, this means that $z \boto w$ is satisfied if and only if $z \cap w \subseteq V$. Thus, $S$ does not satisfy $\{a\} \boto \{a\}$, since by assumption $a \not \in V$.

On the other hand, $S$ satisfies all assertions of $\Sigma$. Indeed, consider any $z \boto w \in \Sigma$. By Decomposition and Symmetry, every element of $c \in M$ which is in both $z$ and $w$ is such that $\Sigma \vdash_I \{c\} \boto \{c\}$, that is, such that $c \in V$. Thus, $z \cap w \subseteq V$ and therefore $z \boto w$ is satisfied by (the relational diversity rank function corresponding to) $S$, as required. 

In conclusion, from the assumption that $\{a\} \boto \{a\}$ does not follow from $\Sigma$ according to the rules we were able to find a relational diversity rank function that satisfies $\Sigma$ but not $\{a\} \boto \{a\}$. Thus (2) implies (3), and this concludes the proof.
\end{proof}

Now we can generalize the completeness result of Lemma \ref{lemma:comp_ind_const} to arbitrary independence assertions. The proof is an adaptation of the proof of \cite{MR1097266}
to our framework.
\begin{theorem}[Completeness of the Independence Axioms]\label{indep}
Let $M$ be a set and let $\Sigma$ be a set of independence assertions $z \boto w$ for $z, w \subseteq_f M$. Then the following conditions are equivalent for any $x, y \subseteq_f M$:\begin{enumerate}

\item  $x\boto y$ holds under {\em every} diversity rank function on $M$ under which $\Sigma$ holds.
\item  $x\boto y$ holds under {\em every} relational diversity rank function   under which $\Sigma$ holds. 
\item $x\boto y$
 follows from $\Sigma$ by the rules of Proposition~\ref{gpp}.
\end{enumerate}
	\label{thm:ind_comp}
\end{theorem}

\begin{proof} We adapt the proof of \cite{MR1097266} to our framework.
Trivially (1) implies (2). Furthermore, (3) implies (1) by Proposition \ref{gpp}. We shall prove that (2) implies (3) by contraposition: supposing that $\Sigma \not \vdash_I x \boto y$, we shall construct a relational diversity rank function that does not satisfy $x \boto y$ but satisfies $\Sigma$. 

Thus, suppose that $\Sigma \not \vdash_I x\boto y$. Without loss of generality, we can assume that $\Sigma$ is closed under the rules.  Note that $x \not = \emptyset$ and $y \not = \emptyset$, since otherwise $x \boto y$ would follow from $\Sigma$ by the Empty Set and Symmetry rules.  We can also assume that $x$ and $y$ are \emph{minimal}, in the sense that if ${x}\hspace{1pt}'\subseteq x$ and ${y}\hspace{1pt}'\subseteq y$ and at least one containment is proper then $\Sigma \vdash_I x' \boto y'$ (if $x$ and $y$ are not minimal in this sense,  we can replace them with minimal subsets $x_\text{min} \subseteq x$ and $y_\text{min} \subseteq y$ such that $\Sigma \not \vdash_I x_\text{min} \boto y_\text{min}$: then any diversity rank function that does not satisfy $x_\text{min} \boto y_\text{min}$ will not satisfy $x \boto y$ either, by the soundness of the Decomposition rule).

If $x = y = c$ for some $c \in M$, the existence of a relational diversity rank function that satisfies $\Sigma$ but not $c \boto c$ follows immediately from Lemma \ref{lemma:comp_ind_const} and the proof is concluded. 

%

Let us suppose instead that this is not the case, and let $V = \{c \in M: \Sigma \vdash_I c \boto c\}$. Notice that $x \cap V = \emptyset$, because if $x = x' c$ for $c \in V$ then by Constancy Augmentation we could derive $x \boto y$ from $x' \boto y$ (which follows from $\Sigma$ because of our minimality assumption) and $c \boto c$. Similarly, $y \cap V = \emptyset$. 

Next, we show that $x \cap y \subseteq V$: since we already proved that $x \cap V = y \cap V = \emptyset$, this will imply at once that $x \cap y = \emptyset$. Let $c \in M$ is such that $c \in x$ and $c \in y$. Since as we said we can exclude the case that $x = y = c$, at least one of them contains another element; and therefore, by minimality, $\Sigma \vdash_I c \boto c$, i.e., $c \in V$. Therefore $x \cap y = \emptyset$, as stated. 

Now, let $S$ be the team with domain $M$ over $\{0,1\}$ consisting of all functions $s: M \rightarrow \{0,1\}$ satisfying the two conditions
\begin{itemize}
    \item $\sum_{a \in x} s(a) \equiv \sum_{b \in y} s(b) \mod 2$; 
    \item $s(c) = 0$ for all $c \not \in xy$
\end{itemize}
and consider the diversity rank function $\rank{\cdot} = \log(\rows_S(\cdot))$ induced by $S$. Recall that, for any two disjoint $z$ and $w$, this diversity rank function satisfies $z \boto w$ if and only $S(zw) = S(z) \times S(w)$, i.e. if and only if all possible values in $S$ for $z$ and for $w$ appear jointly in some row of $S$. 

This diversity rank function does not satisfy $x \boto y$. Indeed, let $a \in x$ and $b \in y$ and consider the two assignments $s_1, s_2 \in S$ defined by 
\begin{itemize}
    \item $s_1(a) = s_1(b) = 1$, $s_1(c) = 0$ for $c \not \in \{a,b\}$; 
    \item $s_2(c) = 0$ for all $c \in M$. 
\end{itemize}
There exists no $s \in S$ that agrees with $s_1$ over $x$ and with $s_2$ over $y$, because this would violate the parity condition; therefore, it is indeed the case that the induced diversity rank function does not satisfy $x \boto y$. 

It remains to show that this induced diversity rank function satisfies all independence assertions in $z \boto w \in \Sigma$ (remember that we assumed that $\Sigma$ is closed under our rules). Let us consider the following cases: 
\begin{enumerate}
   
    \item $z = \emptyset$ or $w = \emptyset$:
    
    By the soundness of Empty Set and Symmetry rules, every diversity rank function (and thus, in particular, our diversity rank function $\rank{\cdot}$ induced by $S$) satisfies $z \boto w$.
    \item $z \cap w \not = \emptyset$: 
    
    By Decomposition and Symmetry, for all $c \in z \cap w$ we have that $\Sigma \vdash_I c \boto c$. Thus, $c \in V$ and therefore (since as we saw $xy \cap V = \emptyset$) $s(c) = 0$ for all $s \in S$. Thus, for $z' = z \backslash w$ and $w' = w \backslash z$, $\rows_S(z) = \rows_S(z')$, $\rows_S(w) = \rows_S(w')$ and $zw = z' w'$. Therefore, $\rank{zw} = \rank{z} + \rank{w}$ if and only if $\rank{z'w'} = \rank{z'} + \rank{w'}$, and therefore it suffices to check whether our diversity rank function satisfies $z' \boto w'$. 
    \item $z \cap w = \emptyset$, $z \not = \emptyset$ and $w \not = \emptyset$: \begin{enumerate}
        \item $z \backslash xy \not = \emptyset$ or $w \backslash xy \not = \emptyset$:
        
        Since all $s \in S$ are such that $s(c) = 0$ for all $c \not \in xy$, the induced diversity rank function $\rank{\cdot}$ satisfies $z \boto w$ if and only if it satisfies $(z \cap xy) \boto (w \cap xy)$. Hence, we need not consider this case further. 
        \item $zw = xy$:
        
        As we will now see, in this case we could conclude that $\Sigma \vdash_I x \boto y$, which contradicts our hypothesis. Thus, this case is not possible. 
        
        Let $x_z = x \cap z$, $x_w = x \cap w$. $y_z = y \cap z$ and $y_w = y \cap w$. Then $z = x_z y_z$, $w = x_w y_w$, $x = x_z x_w$ and $y = y_z y_w$. Suppose $x_z = \emptyset$. Then $x_w \not = \emptyset$ since $x \not = \emptyset$ and $y_z \not = \emptyset$ since $z \not = \emptyset$. Moreover $y_w \not = \emptyset$ for, otherwise, $x = w$ and $y=z$, which is impossible because $\Sigma \not \vdash_I x \boto y$. By symmetry, we conclude that at most one of $x_z$, $x_w$, $y_z$ and $y_w$ is empty. Without loss of generality, assume that both $x_z \not = \emptyset$ and $x_w \not = \emptyset$. By the minimality of $x \boto y$, it follows that $\Sigma \vdash_I x_z \boto y_z$. From this and $\Sigma \vdash_I x_z y_z \boto x_w y_w$, it follows by Mixing that $\Sigma \vdash_I x_z \boto x_w y_z y_w$. Again, by minimality of $x \boto y$. we have that $\Sigma \vdash_I x_w \boto y_z y_w$. Applying Mixing again and Symmetry, it follows that $\Sigma \vdash_I x \boto y$, i.e. $\Sigma \vdash_I x \boto y$. But this contradicts our assumption, and therefore it cannot be the case that $zw = xy$. 
        \item $zw \subsetneq xy$: 
        
        It follows immediately that the relational diversity rank induced by $S$ satisfies $z \boto w$, since we can use the variable(s) in $xy$ but not in $zw$ to tweak the parity condition. 
    \end{enumerate}
\end{enumerate}
Thus, we saw that the diversity rank function induced by $S$ does not satisfy $x \boto y$ but satisfies all independence statements of $\Sigma$, as required. 
\end{proof}

\begin{corollary}[\cite{MR1097266}] An independence assertion  follows semantically, in all 
data\-bases, 
from a given set of independence assertions if and only if it follows by the rules of Proposition~\ref{gpp}.

\end{corollary}


At this point it would be natural to ask the following\\

\noindent\textbf{Open Problem:} What are the rules that govern the interaction between dependence and independence in our framework?\\

In general, inference problems for dependence/independence assertions are not necessarily decidable in a 
relational setting (\cite{herrmann1995undecidability,MR2277338,hannula2018interaction}); but in the setting of general diversity rank functions, the decision problem for sets of dependence and independence assertions is necessarily decidable. Indeed, the axioms of diversity rank functions as well as dependence and independence assertions may be translated into the first-order arithmetic of the reals, replacing each $\rank{x}$ expression with a variable $r_x$ that ranges over non-negative real numbers and -- for example -- writing a dependence atom $=\!\!(x,y)$ as $r_{xy} = r_x$, an independence atom $x \boto y$ as $r_{xy} = r_x + r_y$, and Axiom \textbf{R3} as an axiom schema of the form $(r_{xy} = r_x) \rightarrow (r_{xyz} = r_{xz})$. But the arithmetic of the reals is decidable \cite{MR0028796}, and therefore it is decidable whether some dependence or independence assertion follows from a set $\Sigma$ of dependence and independence assertions in our setting. 

We leave the problem of searching for an axiom system for dependence and independence assertions combined to future work. Here we only point out two simple axioms that govern the interaction between dependence and independence in our setting:
\begin{itemize}
\item \emph{Constancy Equivalence:} $\indepc{ x}{ x}$ if and only if $=\!\!(\emptyset,  x)$; 
\item \emph{Propagation:} If $\indepc{ x}{ y}$ and $=\!\!( y, z)$ then $\indepc{ x}{ y z}$.
\end{itemize}
Both of these can be shown to follow easily from our notion of rank: 
\begin{itemize}
\item \emph{Constancy Equivalence:} Suppose   $\indepc{ x}{ x}$. Then, by definition, $\rank{ x} + \rank{ x} = \rank{ x  x}$. But on the other hand, $x  x = x$ and therefore, $\rank{ x} = 0$ and $\rank{ x} = \rank{\emptyset  x} = \rank{\emptyset} = 0$.
Conversely, suppose   $=\!\!(\emptyset,  x)$. Then $\rank{x} = \rank{\emptyset x} = \rank{\emptyset} = 0$, and therefore $\rank{ x} + \rank{ x} = 0 = \rank{ x  x}$.
\item \emph{Propagation:} Suppose that $\rank{ x} + \rank{ y} = \rank{ x  y}$ and that $\rank{ y  z} = \rank{ y}$. From the second hypothesis, by \textbf{R3}, we can show that $\rank{ x  y} = \rank{ x  y  z}$ and therefore in the first hypothesis we can replace $\rank{ y}$ with $\rank{ y  z}$ and $\rank{ x  y}$ with $\rank{ x  y  z}$, thus obtaining $\rank{ x} + \rank{ y  z} = \rank{ x  y  z}$.
Therefore $\indepc{ x}{ y  z}$, as required.
\end{itemize}
\section{A Representation Theorem for Dependence Atoms}
As we saw, every dependence notion induced by a diversity rank function satisfies Armstrong's Axioms, which furthermore are complete for diversity rank functions. However, a question that remains open is whether every dependency notion that satisfies Armstrong's Axioms is induced by a diversity rank function. 

More formally, let $M$ be a set, and let $\Sigma$ be a set of dependence assertions
over $M$ which is closed under Armstrong’s Axioms. Is there a
diversity rank function $\rank{\cdot}$ for which $\Sigma = \{=\!\!(x, y) : x, y \subseteq_f M, \rank{xy} = \rank{x}\}$?\footnote{We thank Samson Abramsky for asking this question in a personal communication.}
%

An example may be helpful here. Suppose that $M = \{a,b,c\}$ and $\Sigma$ is the closure under Armstrong's Axioms of $\{=\!\!(ab, c)\}$, i.e.
\[
\Sigma = \{=\!\!(x, y): x, y \subseteq M, y \subseteq x\} \cup \{=\!\!(x, y): ab \subseteq x \text{ and } y \subseteq M\}.
\]
It is not hard to verify that this set is indeed closed under Reflexivity, Augmentation and Transitivity; but can we find a diversity rank function over $\{a,b,c\}$ such that, for $x, y \subseteq \{a,b,c\}$, $y$ depends on $x$ according to it if and only if $=\!\!(x,y) \in \Sigma$?

By Theorem \ref{thm:represent} below, such a diversity function does exist.
\begin{theorem}
Let $M$ be a countable set and let $\Sigma$ be a set of dependency assertions of the form $\fdep(x, y)$ (for $x, y \subseteq_f M$) that is closed under Armstrong's Axioms. Then there exists a diversity rank function $\rank{\cdot}$ over $M$ such that $\Sigma_{\rank{\cdot}} = \Sigma$. 
\label{thm:represent}
\end{theorem}
\begin{proof}
We first define the relation $\equiv$ on the finite subsets of $M$, as follows. For $x, y \subseteq_f M$, $x\equiv y$ if both $\fdep(x, y)\in\Sigma$ and $\fdep(y, x)\in\Sigma$.
Using that $\Sigma$ satisfies Armstrong's axioms, it is straightforwardly seen
that $\equiv$ is an equivalence relation. For $x \subseteq_f M$, let $E_x$ be the equivalence
class of $x$, and let $\mathcal{E} = \{E_x : x \subseteq_f M\}$. We first observe that if $E_{x_1} = E_{x_2}$
and $E_{y_1} = E_{y_2}$, then $\fdep(x_1, y_1)\in\Sigma$ if and only if  $\fdep(x_2, y_2)\in\Sigma $. To see this,
it suffices to prove one direction, by symmetry. Thus suppose $\fdep(x_1, y_1)\in\Sigma$.
Since $E_{x_1} = E_{x_2}$, $\fdep(x_2, x_1)\in\Sigma$. Since $E_{y_1} = E_{y_2}$, $\fdep(y_1, y_2)\in\Sigma$.
By two applications of Transitivity, it now follows that $\fdep(x_2, y_2)\in\Sigma$.
This observation justifies the following definition of the relation $\le$ on
$\mathcal{E}$: for $x, y \subseteq_f M, E_y\le E_x$ if $\fdep(x, y)\in\Sigma$. It is straightforwardly seen
that $\le$ is a partial order on $\mathcal{E}$: reflexivity and transitivity follow from
$\Sigma$ satisfying Reflexivity and Transitivity, and antisymmetry follows
from the definition of the equivalence relation $\equiv$.

Next, we observe that the set of all finite subsets of $M$ is countable,
because $M$ is.
Hence $\mathcal{E}$ is also countable. 
Let $\mathcal{E} = \{E_0,E_1,E_2,E_3, . . .\}$
be an enumeration of $\mathcal{E}$ with $E_0 = E_\emptyset$. We now construct an order-preserving injection $f : \mathcal{E} \to \{0\}\cup ]1, 2[$ recursively as follows:
\begin{enumerate}
\item $f(E_0) = 0$.
\item Assume that, for $i = 0, . . . , n$, $f(E_i)$ has been defined. Let
$$l = \max(\{f(E_i) : 1\le i\le n, E_i < E_{n+1}\} \cup \{1\}),$$
$$r = \min(\{f(E_i) :1\le i\le n, E_i > E_{n+1}\} \cup \{2\}).$$
Then, let $f(E_{n+1})$ be any value in $]l, r[\setminus\{f(E_1), . . . , f(E_n)\}$. 
\end{enumerate}
Notice
that, for all $x \subseteq_f M$, $E_0 = E_\emptyset\le E_x$, since $\fdep(x, \emptyset)\in\Sigma$ by Reflexivity. Hence, the construction of $f$ guarantees that it is indeed
order-preserving.

Finally, we define a diversity rank function $\rank{\cdot}$ on M as follows: for
$x\subseteq_f M$, $\rank{x} = f(E_x)$. We prove that $\rank{\cdot}$ satisfies R1-R4.
\begin{description}
\item [R1:] We have that $\rank{\emptyset} = f(E_\emptyset) = 0$.
\item [R2:] Let $x, y \subseteq_f M$. Since $x\subseteq xy, \fdep(xy, x)\in\Sigma$ by Reflexivity.
Hence $E_x\le E_{xy}$. Since $f$ is order-preserving, $\rank{x} = f(E_x) \le
f(E_{xy}) = \rank{xy}$. If $x = \emptyset$, or $y = \emptyset$, then $\rank{xy} \le \rank{x} + \rank{y}$ is
satisfied by R1; if $x \ne\emptyset$, and $y \ne\emptyset$, then $\rank{xy} \le \rank{x} + \rank{y}$ is
satisfied since, by construction of $f$, $1 < \rank{x}, \rank{y}, \rank{xy} < 2$.
\item [R3:] Let $x, y \subseteq_f M$ and assume that $\rank{x} = \rank{xy}$. Then, $f(E_x) =
f(E_{xy})$, hence $E_x = E_{xy}$ since $f$ is injective. If follows that both
$\fdep(xy, x)\in\Sigma$ and $\fdep(x, xy)\in\Sigma$.\footnote{The former is of course always the case, because of Reflexivity.}
 Now, let $z \subseteq_f M$. By Augmentation, we have that both $\fdep(xyz, xz)\in\Sigma$ and $\fdep(xy, xyz)\in\Sigma$. Hence
$E_{xy} = E_{xyz}$. It follows that $\rank{xy} = f(E_{xy}) = f(E_{xyz}) = \rank{xyz}$.
\item [R4:] Notice that this property is satisfied by R1 if $x = \emptyset$ or $y =\emptyset$, and is trivially satisfied if $z = \emptyset$. If $x \ne\emptyset, y \ne\emptyset$, and $z \ne\emptyset$, the property is satisfied because the antecedent is false, as, by
construction of $f$, $1 < \rank{xyz}, \rank{x}, \rank{yz} < 2$.
\end{description}
It remains to show that $\Sigma= \{\fdep(x, y) : x, y \in M, \rank{xy} = \rank{x}\}$. If
$\fdep(x, y)\in\Sigma$, then by Augmentation, $\fdep(x, xy)\in\Sigma$. By Reflexivity,
$\fdep(xy, x)\in\Sigma$. Hence, $E_x = E_{xy}$ and $\rank{x} = f(E_x) = f(E_{xy}) = \rank{xy}$.
Conversely, if $\rank{x} = \rank{xy}$, then $f(E_x) = f(E_{xy})$. Hence, $E_x = E_{xy}$,
since $f$ is injective. Consequently, $\fdep(x, xy)\in\Sigma$, and, by Reflexivity
and Transitivity, $\fdep(x, y)\in\Sigma$.
\end{proof}

As an example, let us apply the procedure described here to $M = \{a,b,c\}$ and to the set $\Sigma$ we mentioned before, i.e. 
\[
\Sigma = \{\fdep(x, y): x, y \subseteq M, y \subseteq x\} \cup \{\fdep(x, y): ab \subseteq x \text{ and } y \subseteq M\}.
\]

It is readily seen that $E_{ab} = E_{abc} = \{ab, abc\}$. For all other $x \subseteq_f M$,
$E_x = \{x\}$. Let us now consider the following enumeration of the finite
subsets of $\mathcal{E} = \{E_x : x \subseteq_f M\}$:
$$E_\emptyset,E_{a},E_{b},E_{c},E_{ab} = E_{abc},E_{ac},E_{bc}.$$
We now construct the function $f$ in the proof of Theorem 3 recursively
according to this enumeration:
\begin{enumerate}
\item $f(\emptyset) = 0$, by construction.
\item  Since $E_{a}, E_{b}$, and $E_{c}$ are mutually incomparable, we may assign
to them arbitrary mutually different values between $1$ and $2$, e.g.,
$f(E_{a}) = 1.5, f(E_{b}) = 1.6$, and $f(E_{c}) = 1.1$.
\item  Since, for every $x \subseteq_f M, \fdep(ab, x)\in\Sigma$, we have in particular that
$E_{a} < E_{ab}, E_{b} < E_{ab}$, and $E_{c} < E_{ab}$. Therefore, we must assign to $E_{ab}$ a value between $1.6$ and $2$, say, $f(E_{ab}) = 1.8$.
\item  Since $\fdep(ac,a)\in\Sigma, \fdep(ac,c)\in\Sigma$, and $\fdep(ab, ac)\in\Sigma$, we have $E_{a} < E_{ac},
E_{c} < E_{ac}$, and $E_{ac} < E_{ab}$. Since $\fdep(ac, b) \notin\Sigma$ and $\fdep(b, ac) \notin\Sigma$,
$E_{b}$ and $E_{ac}$ are incomparable. Therefore, we must assign to $E_{ac}$ an
unused value between $1.5$ and $1.8$, say, $f(E_{ac}) = 1.7$.
\item  Since $\fdep(bc,b)\in\Sigma, \fdep(bc,c)\in\Sigma$, and $\fdep(ab, bc)\in\Sigma$, $E_{b} < E_{bc}$,
$E_{c} < E_{bc}$, and $E_{bc} < E_{ab}$. Since $\fdep(bc, a) \notin\Sigma$  and $\fdep(a, bc) \notin\Sigma, E_{a}$
and $E_{bc}$ are incomparable. Since $\fdep(bc, ac) \notin\Sigma$ and $\fdep(ac, bc) \notin\Sigma$,
$E_{ac}$ and $E_{bc}$ are incomparable. Therefore, we must assign to $E_{bc}$ an
unused value between $1.6$ and $1.8$, say, $f(E_{bc}) = 1.65$.
\end{enumerate}

According to the proof of Theorem \ref{thm:represent}, $\Sigma = \{=\!\!(x, y) : x, y \subseteq_f M, \rank{xy} = \rank{x}\}$, as can be easily verified. It is also easily verified that this is indeed a diversity rank function. 
%

We leave for future work the question whether this representation theorem may be generalized to independence atoms, that is to say,  whether every set of dependence \emph{and independence} assertions satisfying the axioms for dependence, the ones for independence, and some additional axioms for dependence/independence interactions arises from some diversity rank function. 

\section{Conclusions} 
In this work, we showed how many distinct notions of dependence and independence, originating in different branches of mathematics and computer science, may be treated as instances of the same framework: one which can be seen as a generalization of matroid theory allowing for non-integer ranks and weakening the submodularity condition. In this framework, $y$ is said to be dependent from $x$ if adding it to $x$ does not increase  the amount of diversity, while $y$ is said to be independent from $x$ if adding it to $x$ increases \emph{maximally} the amount of diversity. 

Despite its generality, this framework is nonetheless powerful enough to prove non-trivial results - including, in particular, completeness theorems for the corresponding dependence and independence notions. These results generalize to our entire setting the completeness theorems by Armstrong and by Geiger-Paz-Pearl for database-theoretic functional dependence and for probabilistic independence respectively.  In addition, we have obtained  a representation theorem showing that every set of dependence assertions that satisfies Armstrong's Axioms arises from some diversity rank function. 

One natural next step would be to investigate further the properties of this formalism, in particular with respect to the interaction between independence and dependence assertions. Combinatorial properties of this system would also be worth investigating, as would  the study of possible operations that \emph{combine} different diversity rank functions. This could also contribute to the logical study of notions of dependence and independence in the context of Team Semantics, in particular providing a unifying approach for its different variants (e.g. probabilistic, modal, propositional, \ldots). 

Another important next step would be to extend our representation result to dependence \emph{and independence} assertions combined, as well as investigate the connections between our approach and other approaches, such as the study of measure-based constraints in the context of Database Theory  \cite{sayrafi2008implication}, and  the lattice-theoretic  study of conditional independence of \cite{niepert2013conditional}. 

Finally, it would be interesting to investigate potential applications of our framework. Our approach in this work has been one of synthesis, motivated by the search of a framework that captures a wide array of notions of dependence and independence studied in different areas. Such a framework had to be more general than matroids, which -- despite their success and undeniable importance -- are limited by their commitment to submodularity and to integer values; had to be formally simple and widely applicable; and it had to be specific enough to highlight true commonalities and analogies between notions of dependence and independence developed in different areas. We think that our notion of diversity rank function meets these objectives, and that this has the potential to lead to interesting connections and fruitful cross-topic 
fertilization between the study of notions of dependence and independence in these different areas. 

\bigskip

\noindent{\bf Acknowledgments:}

We thank the reviewers for a number of helpful comments and suggestions. We are particularly grateful to Reviewer 2 for their careful reading of the manuscript and for providing a clearer proof for Theorem 3.
The second  author would like to thank  the Academy of Finland, grant no: 322795. This project has received funding from the European Research Council (ERC) under the European Union’s Horizon 2020 research and innovation programme (grant agreement No 101020762).



\begin{thebibliography}{29}
\ifx \bisbn   \undefined \def \bisbn  #1{ISBN #1}\fi
\ifx \binits  \undefined \def \binits#1{#1}\fi
\ifx \bauthor  \undefined \def \bauthor#1{#1}\fi
\ifx \batitle  \undefined \def \batitle#1{#1}\fi
\ifx \bjtitle  \undefined \def \bjtitle#1{#1}\fi
\ifx \bvolume  \undefined \def \bvolume#1{\textbf{#1}}\fi
\ifx \byear  \undefined \def \byear#1{#1}\fi
\ifx \bissue  \undefined \def \bissue#1{#1}\fi
\ifx \bfpage  \undefined \def \bfpage#1{#1}\fi
\ifx \blpage  \undefined \def \blpage #1{#1}\fi
\ifx \burl  \undefined \def \burl#1{\textsf{#1}}\fi
\ifx \doiurl  \undefined \def \doiurl#1{\url{https://doi.org/#1}}\fi
\ifx \betal  \undefined \def \betal{\textit{et al.}}\fi
\ifx \binstitute  \undefined \def \binstitute#1{#1}\fi
\ifx \binstitutionaled  \undefined \def \binstitutionaled#1{#1}\fi
\ifx \bctitle  \undefined \def \bctitle#1{#1}\fi
\ifx \beditor  \undefined \def \beditor#1{#1}\fi
\ifx \bpublisher  \undefined \def \bpublisher#1{#1}\fi
\ifx \bbtitle  \undefined \def \bbtitle#1{#1}\fi
\ifx \bedition  \undefined \def \bedition#1{#1}\fi
\ifx \bseriesno  \undefined \def \bseriesno#1{#1}\fi
\ifx \blocation  \undefined \def \blocation#1{#1}\fi
\ifx \bsertitle  \undefined \def \bsertitle#1{#1}\fi
\ifx \bsnm \undefined \def \bsnm#1{#1}\fi
\ifx \bsuffix \undefined \def \bsuffix#1{#1}\fi
\ifx \bparticle \undefined \def \bparticle#1{#1}\fi
\ifx \barticle \undefined \def \barticle#1{#1}\fi
\ifx \bconfdate \undefined \def \bconfdate #1{#1}\fi
\ifx \botherref \undefined \def \botherref #1{#1}\fi
\ifx \url \undefined \def \url#1{\textsf{#1}}\fi
\ifx \bchapter \undefined \def \bchapter#1{#1}\fi
\ifx \bbook \undefined \def \bbook#1{#1}\fi
\ifx \bcomment \undefined \def \bcomment#1{#1}\fi
\ifx \oauthor \undefined \def \oauthor#1{#1}\fi
\ifx \citeauthoryear \undefined \def \citeauthoryear#1{#1}\fi
\ifx \endbibitem  \undefined \def \endbibitem {}\fi
\ifx \bconflocation  \undefined \def \bconflocation#1{#1}\fi
\ifx \arxivurl  \undefined \def \arxivurl#1{\textsf{#1}}\fi
\csname PreBibitemsHook\endcsname

\bibitem{MR1507091}
\begin{barticle}
\bauthor{\bsnm{Whitney}, \binits{H.}}:
\batitle{On the {A}bstract {P}roperties of {L}inear {D}ependence}.
\bjtitle{Amer. J. Math.}
\bvolume{57}(\bissue{3}),
\bfpage{509}--\blpage{533}
(\byear{1935}).
\doiurl{10.2307/2371182}
\end{barticle}
\endbibitem

\bibitem{MR0002841}
\begin{bbook}
\bauthor{\bparticle{van~der} \bsnm{Waerden}, \binits{B.L.}}:
\bbtitle{Moderne {A}lgebra},
p. \bfpage{224}.
\bpublisher{J. Springer},
\blocation{Berlin}
(\byear{1940})
\end{bbook}
\endbibitem

\bibitem{armstrong}
\begin{botherref}
\oauthor{\bsnm{Armstrong}, \binits{W.W.}}:
Dependency structures of data base relationships.
Information Processing
\textbf{74}
(1974)
\end{botherref}
\endbibitem

\bibitem{MR2351449}
\begin{bbook}
\bauthor{\bsnm{V{\"a}{\"a}n{\"a}nen}, \binits{J.}}:
\bbtitle{Dependence Logic}.
\bsertitle{London Mathematical Society Student Texts},
vol. \bseriesno{70},
p. \bfpage{225}.
\bpublisher{Cambridge University Press},
\blocation{Cambridge}
(\byear{2007})
\end{bbook}
\endbibitem

\bibitem{MR2807973}
\begin{bbook}
\bauthor{\bsnm{Mann}, \binits{A.L.}},
\bauthor{\bsnm{Sandu}, \binits{G.}},
\bauthor{\bsnm{Sevenster}, \binits{M.}}:
\bbtitle{Independence-friendly Logic}.
\bsertitle{London Mathematical Society Lecture Note Series},
vol. \bseriesno{386},
p. \bfpage{208}.
\bpublisher{Cambridge University Press},
\blocation{Cambridge}
(\byear{2011}).
\doiurl{10.1017/CBO9780511981418}.
\bcomment{A game-theoretic approach}.
\burl{http://dx.doi.org/10.1017/CBO9780511981418}
\end{bbook}
\endbibitem

\bibitem{GV}
\begin{barticle}
\bauthor{\bsnm{Gr\"{a}del}, \binits{E.}},
\bauthor{\bsnm{V\"{a}\"{a}n\"{a}nen}, \binits{J.}}:
\batitle{Dependence and independence}.
\bjtitle{Studia Logica}
\bvolume{101}(\bissue{2}),
\bfpage{399}--\blpage{410}
(\byear{2013}).
\doiurl{10.1007/s11225-013-9479-2}
\end{barticle}
\endbibitem

\bibitem{yang2016propositional}
\begin{barticle}
\bauthor{\bsnm{Yang}, \binits{F.}},
\bauthor{\bsnm{V{\"a}{\"a}n{\"a}nen}, \binits{J.}}:
\batitle{Propositional logics of dependence}.
\bjtitle{Annals of Pure and Applied Logic}
\bvolume{167}(\bissue{7}),
\bfpage{557}--\blpage{589}
(\byear{2016})
\end{barticle}
\endbibitem

\bibitem{vaananen2008modal}
\begin{bchapter}
\bauthor{\bsnm{V\"{a}\"{a}n\"{a}nen}, \binits{J.}}:
\bctitle{Modal dependence logic}.
In: \bbtitle{New Perspectives on Games and Interaction}.
\bsertitle{Texts in Logic and Games},
vol. \bseriesno{4},
pp. \bfpage{237}--\blpage{254}.
\bpublisher{Amsterdam Univ. Press},
\blocation{Amsterdam}
(\byear{2008})
\end{bchapter}
\endbibitem

\bibitem{hella2014expressive}
\begin{bchapter}
\bauthor{\bsnm{Hella}, \binits{L.}},
\bauthor{\bsnm{Luosto}, \binits{K.}},
\bauthor{\bsnm{Sano}, \binits{K.}},
\bauthor{\bsnm{Virtema}, \binits{J.}}:
\bctitle{The expressive power of modal dependence logic}.
In: \bbtitle{Advances in Modal Logic. {V}ol. 10},
pp. \bfpage{294}--\blpage{312}.
\bpublisher{Coll. Publ.},
\blocation{London}
(\byear{2014})
\end{bchapter}
\endbibitem

\bibitem{krebs2017team}
\begin{bchapter}
\bauthor{\bsnm{Krebs}, \binits{A.}},
\bauthor{\bsnm{Meier}, \binits{A.}},
\bauthor{\bsnm{Virtema}, \binits{J.}},
\bauthor{\bsnm{Zimmermann}, \binits{M.}}:
\bctitle{Team semantics for the specification and verification of
  hyperproperties}.
In: \bbtitle{43rd {I}nternational {S}ymposium on {M}athematical {F}oundations
  of {C}omputer {S}cience}.
\bsertitle{LIPIcs. Leibniz Int. Proc. Inform.},
vol. \bseriesno{117},
pp. \bfpage{10}--\blpage{16}.
\bpublisher{Schloss Dagstuhl. Leibniz-Zent. Inform.},
\blocation{Wadern}
(\byear{2018})
\end{bchapter}
\endbibitem

\bibitem{hyttinen2017logic}
\begin{barticle}
\bauthor{\bsnm{Hyttinen}, \binits{T.}},
\bauthor{\bsnm{Paolini}, \binits{G.}},
\bauthor{\bsnm{V{\"a}{\"a}n{\"a}nen}, \binits{J.}}:
\batitle{A logic for arguing about probabilities in measure teams}.
\bjtitle{Archive for Mathematical Logic}
\bvolume{56}(\bissue{5-6}),
\bfpage{475}--\blpage{489}
(\byear{2017})
\end{barticle}
\endbibitem

\bibitem{durand2018probabilistic}
\begin{bchapter}
\bauthor{\bsnm{Durand}, \binits{A.}},
\bauthor{\bsnm{Hannula}, \binits{M.}},
\bauthor{\bsnm{Kontinen}, \binits{J.}},
\bauthor{\bsnm{Meier}, \binits{A.}},
\bauthor{\bsnm{Virtema}, \binits{J.}}:
\bctitle{Probabilistic team semantics}.
In: \bbtitle{Foundations of Information and Knowledge Systems}.
\bsertitle{Lecture Notes in Comput. Sci.},
vol. \bseriesno{10833},
pp. \bfpage{186}--\blpage{206}.
\bpublisher{Springer},
\blocation{Cham}
(\byear{2018}).
\doiurl{10.1007/978-3-319-90050-6_1}.
\burl{https://doi.org/10.1007/978-3-319-90050-6_1}
\end{bchapter}
\endbibitem

\bibitem{hannula2019facets}
\begin{bchapter}
\bauthor{\bsnm{Hannula}, \binits{M.}},
\bauthor{\bsnm{Hirvonen}, \binits{{\AA}.}},
\bauthor{\bsnm{Kontinen}, \binits{J.}},
\bauthor{\bsnm{Kulikov}, \binits{V.}},
\bauthor{\bsnm{Virtema}, \binits{J.}}:
\bctitle{Facets of distribution identities in probabilistic team semantics}.
In: \bbtitle{Logics in Artificial Intelligence}.
\bsertitle{Lecture Notes in Comput. Sci.},
vol. \bseriesno{11468},
pp. \bfpage{304}--\blpage{320}.
\bpublisher{Springer},
\blocation{Cham}
(\byear{2019}).
\doiurl{10.1007/978-3-030-19570-0_2}.
\burl{https://doi.org/10.1007/978-3-030-19570-0_2}
\end{bchapter}
\endbibitem

\bibitem{oxley2006matroid}
\begin{bbook}
\bauthor{\bsnm{Oxley}, \binits{J.G.}}:
\bbtitle{Matroid Theory}
vol. \bseriesno{3}.
\bpublisher{Oxford University Press},
\blocation{USA}
(\byear{2006})
\end{bbook}
\endbibitem

\bibitem{10.2307/2333344}
\begin{barticle}
\bauthor{\bsnm{Good}, \binits{I.J.}}:
\batitle{The population frequencies of species and the estimation of population
  parameters}.
\bjtitle{Biometrika}
\bvolume{40}(\bissue{3/4}),
\bfpage{237}--\blpage{264}
(\byear{1953})
\end{barticle}
\endbibitem

\bibitem{Hill}
\begin{barticle}
\bauthor{\bsnm{Hill}, \binits{M.}}:
\batitle{Diversity and evenness: A unifying notation and its consequences}.
\bjtitle{Ecology}
\bvolume{54},
\bfpage{427}--\blpage{432}
(\byear{1973}).
\doiurl{10.2307/1934352}
\end{barticle}
\endbibitem

\bibitem{magurran2004measuring}
\begin{bbook}
\bauthor{\bsnm{Magurran}, \binits{A.E.}}:
\bbtitle{Measuring Biological Diversity}.
\bpublisher{Wiley},
\blocation{New Jersey}
(\byear{2004}).
\burl{https://books.google.fi/books?id=tUqzLSUzXxcC}
\end{bbook}
\endbibitem

\bibitem{borda2011fundamentals}
\begin{bbook}
\bauthor{\bsnm{Borda}, \binits{M.}}:
\bbtitle{Fundamentals in Information Theory and Coding}.
\bpublisher{Springer},
\blocation{Berlin, Heidelberg}
(\byear{2011})
\end{bbook}
\endbibitem

\bibitem{cover1991entropy}
\begin{barticle}
\bauthor{\bsnm{Cover}, \binits{T.M.}},
\bauthor{\bsnm{Thomas}, \binits{J.A.}}:
\batitle{Entropy, relative entropy and mutual information}.
\bjtitle{Elements of information theory}
\bvolume{2},
\bfpage{1}--\blpage{55}
(\byear{1991})
\end{barticle}
\endbibitem

\bibitem{MR1097266}
\begin{barticle}
\bauthor{\bsnm{Geiger}, \binits{D.}},
\bauthor{\bsnm{Paz}, \binits{A.}},
\bauthor{\bsnm{Pearl}, \binits{J.}}:
\batitle{Axioms and algorithms for inferences involving probabilistic
  independence}.
\bjtitle{Inform. and Comput.}
\bvolume{91}(\bissue{1}),
\bfpage{128}--\blpage{141}
(\byear{1991}).
\doiurl{10.1016/0890-5401(91)90077-F}
\end{barticle}
\endbibitem

\bibitem{shannon1948mathematical}
\begin{barticle}
\bauthor{\bsnm{Shannon}, \binits{C.E.}}:
\batitle{A mathematical theory of communication}.
\bjtitle{The Bell system technical journal}
\bvolume{27}(\bissue{3}),
\bfpage{379}--\blpage{423}
(\byear{1948})
\end{barticle}
\endbibitem

\bibitem{codd1970relational}
\begin{botherref}
\oauthor{\bsnm{Codd}, \binits{E.F.}}:
A relational model of data for large shared data banks communications.
Communications of the ACM
\textbf{26}(1)
(1970)
\end{botherref}
\endbibitem

\bibitem{lang2002algebra}
\begin{bbook}
\bauthor{\bsnm{Lang}, \binits{S.}}, \betal:
\bbtitle{Algebra}.
\bpublisher{Springer},
\blocation{Berlin, Heidelberg}
(\byear{2002})
\end{bbook}
\endbibitem

\bibitem{herrmann1995undecidability}
\begin{barticle}
\bauthor{\bsnm{Herrmann}, \binits{C.}}:
\batitle{On the undecidability of implications between embedded multivalued
  database dependencies}.
\bjtitle{Information and Computation}
\bvolume{122}(\bissue{2}),
\bfpage{221}--\blpage{235}
(\byear{1995})
\end{barticle}
\endbibitem

\bibitem{MR2277338}
\begin{barticle}
\bauthor{\bsnm{Herrmann}, \binits{C.}}:
\batitle{Corrigendum to: ``{O}n the undecidability of implications between
  embedded multivalued database dependencies'' [{I}nform. and {C}omput. {\bf
  122} (1995), no. 2, 221--235; mr1358026]}.
\bjtitle{Inform. and Comput.}
\bvolume{204}(\bissue{12}),
\bfpage{1847}--\blpage{1851}
(\byear{2006}).
\doiurl{10.1016/j.ic.2006.09.002}
\end{barticle}
\endbibitem

\bibitem{hannula2018interaction}
\begin{bchapter}
\bauthor{\bsnm{Hannula}, \binits{M.}},
\bauthor{\bsnm{Link}, \binits{S.}}:
\bctitle{On the interaction of functional and inclusion dependencies with
  independence atoms}.
In: \bbtitle{International Conference on Database Systems for Advanced
  Applications},
pp. \bfpage{353}--\blpage{369}
(\byear{2018}).
\bcomment{Springer}
\end{bchapter}
\endbibitem

\bibitem{MR0028796}
\begin{bbook}
\bauthor{\bsnm{Tarski}, \binits{A.}}:
\bbtitle{A {D}ecision {M}ethod for {E}lementary {A}lgebra And {G}eometry},
p. \bfpage{60}.
\bpublisher{RAND Corporation},
\blocation{Santa Monica, Calif.}
(\byear{1948})
\end{bbook}
\endbibitem

\bibitem{sayrafi2008implication}
\begin{barticle}
\bauthor{\bsnm{Sayrafi}, \binits{B.}},
\bauthor{\bsnm{Van~Gucht}, \binits{D.}},
\bauthor{\bsnm{Gyssens}, \binits{M.}}:
\batitle{The implication problem for measure-based constraints}.
\bjtitle{Information Systems}
\bvolume{33}(\bissue{2}),
\bfpage{221}--\blpage{239}
(\byear{2008})
\end{barticle}
\endbibitem

\bibitem{niepert2013conditional}
\begin{barticle}
\bauthor{\bsnm{Niepert}, \binits{M.}},
\bauthor{\bsnm{Gyssens}, \binits{M.}},
\bauthor{\bsnm{Sayrafi}, \binits{B.}},
\bauthor{\bsnm{Van~Gucht}, \binits{D.}}:
\batitle{On the conditional independence implication problem: A
  lattice-theoretic approach}.
\bjtitle{Artificial Intelligence}
\bvolume{202},
\bfpage{29}--\blpage{51}
(\byear{2013})
\end{barticle}
\endbibitem

\end{thebibliography}

\def\Dbar{\leavevmode\lower.6ex\hbox to 0pt{\hskip-.23ex \accent"16\hss}D}
  \def\cprime{$'$}

\end{document}